\documentclass[11pt,a4paper,twoside]{article}

\usepackage[utf8]{inputenc}
\usepackage[margin=2.5cm]{geometry}
\usepackage[shortlabels]{enumitem}
\usepackage[page,toc,title]{appendix}
\usepackage{amsmath,amssymb,graphicx,xcolor,subcaption,float,amsthm,mathrsfs,mathtools,bbold,framed,url,amscd,soul,mathabx,tikz,tikz-cd}
\usepackage[color=yellow]{todonotes}
\usepackage{comment}
\usepackage{mathtools}
\usepackage{graphicx}
\setstcolor{red}
\usepackage{cancel}

\theoremstyle{plain}
\newcommand{\thistheoremname}{}
\newtheorem*{genericthm*}{\thistheoremname}
\newenvironment{namedthm*}[1]
  {\renewcommand{\thistheoremname}{#1}%
   \begin{genericthm*}}
  {\end{genericthm*}}

\usepackage{hyperref}
\hypersetup{
colorlinks=true,
}
\numberwithin{equation}{section}

\newcommand{\nocontentsline}[3]{}
\newcommand{\tocless}[2]{\bgroup\let\addcontentsline=\nocontentsline#1{#2}\egroup}

\newtheorem{proposition}{Proposition}

\theoremstyle{definition}

\theoremstyle{remark}
\newtheorem{remark}{Remark}[section]

\usepackage{fancyhdr}
\emergencystretch=\maxdimen
\newcommand{\bs}[1]{\boldsymbol{#1}}

\def\p{{\partial}}
\def\zh{{\bs{\widehat{z}}}}

\newcommand{\ob}[1]{\overline{#1}}

\newcommand{\mc}[1]{\mathcal{#1}}

\newcommand{\mcal}[1]{\mc{#1}}

\def\zh{{\bs{\widehat{z}}}}

\def\p{{\partial}}

\def\bR{{\boldsymbol{R}}}
\def\bu{{\boldsymbol{u}}}

\def\bx{{\boldsymbol{x}}}

\begin{document}
	\title{A thermal Green-Naghdi model with time dependent bathymetry and complete Coriolis force}
	\author{
    Darryl D. Holm\thanks{Department of Mathematics, Imperial College London} \and Oliver D. Street\thanks{Grantham Institute, Imperial College London, email: o.street18\@imperial.ac.uk }}
	\date{\today}
	\maketitle
	\begin{abstract}
        This paper extends the theoretical Euler-Poincar\'e framework for modelling ocean mixed layer dynamics. Through a symmetry-broken Lie group invariant variational principle, we derive a generalised Green-Naghdi equation with time dependent bathymetry, a complete Coriolis force, and inhomogeneity of the thermal buoyancy. The nature of the model derived here lends it a potential future application to wave dynamics generated by changes to the bathymetry.
	\end{abstract}

\tableofcontents

\section{Introduction} \label{Intro}

The Green-Naghdi model is a two dimensional geophysical fluid model for the dynamics of a shallow fluid with free upper boundary. The lower boundary of the domain is usually taken to be the (fixed) bathymetry below the fluid. In reduced barotropic-mode models (also called $1.5$ layer models) for oceanic mixed layer dynamics, the upper boundary of the domain is a free surface and the lower boundary is taken to be the thermocline. In these models, the mixed layer of the ocean comprises the transitional state of fluid dynamics coupling the wind forcing at the ocean free surface and the oscillations of the thermocline below, riding the passive stably stratified depths of the deep ocean. With the assumption of a dynamically passive deep lower layer, the thermocline moves exactly out of phase with the upper free surface. In keeping with the dynamically passive lower layer assumption, most of this variation results from the deformation of the bottom boundary (thermocline). Thus, this presents the modelling challenge of identifying the influence of introducing time dependence into the `bathymetry' on the properties and geometric structure of the Green-Naghdi model.

These large movements of the thermocline give rise to non-hydrostatic pressure terms which, in vertically-integrated shallow-water models such as the Green-Naghdi equations, show up as dispersive terms involving higher-order derivatives of the velocity \cite{GN, BMP, DH88, NMS}. Such nonlinearly dispersive higher-order terms are known to matter dramatically in shallow water flow over an obstacle, \cite{NMS}. Nonetheless, these terms are usually neglected in the hydrostatic approximations, although their nonlinear dispersive effects on the low frequency variability may sometimes be significant. It should be noted that Green-Naghdi in particular is not advised in the multi-layer configuration, because the multi-layer Green--Naghdi equations are ill-posed \cite{CHP2010}.

    Traditional approximated models of geophysical flow involve assumptions due to the aspect ratio of the domain and flow. When approximating a full rotating fluid model in spherical geometry, it is standard to assume that the locally horizontal components of the rotation vector are sufficiently small to be ignored. This approximation can be poor when horizontal length scales are sufficiently well resolved, or when the flow has significant dynamics in the vertical direction. A correction to this can be made by reintroducing horizontal components into the rotation vector \cite{DS05,SD2010}, thereby considering the complete Coriolis force in an approximate two dimensional model.

    In this work, following \cite{HolmLuesink2021}, we will also be introducing thermal effects into the Green-Naghdi model. Modelling thermal fluids has a long and diverse history. There is naturally a particular interest within the geophysical sciences, and we here bring attention to the development of two dimensional thermal models of fluid flow in a shallow domain. Such models were notably developed by Ripa \cite{Ripa93,Ripa95}, who contributed to the thermal shallow water equation (sometimes referred to as the `Ripa model', or ${\rm IL}^{0}$ in Ripa's notation), later considered further in \cite{Dellar2003} and \cite{WaDe2013}. The thermal shallow water equation was extended in \cite{BV21a} to include polynomial (of degree $\alpha$) stratification in the vertical direction (the ${\rm IL}^{0,\alpha}$ model) and, in a separate work \cite{BV2021}, to the multilayer setting. See \cite{CHP2010} for a discussion of the Green-Naghdi model in a multilayer context. Recent interest in thermal ocean modelling has involved the consideration of stochastic parameterisation and data assimilation methodologies \cite{HLP2021}, the variational structure of the thermal vertical slice model attributed to Eady \cite{HHS2024a}, and the study of wave-current interaction in a thermal fluid \cite{HHS2024b}.
    
In this work we formulate the theoretical framework for comparing the solution behaviour of vertically integrated shallow water models, using traditional Coriolis force and fixed bathymetry, with the solution behaviour of a new extended Green-Naghdi model. The new model presented here is a variant of the thermal rotating Green--Naghdi equations, using a complete Coriolis force and featuring interactions between the fluid and a moving bathymetry representing the variability of the thermocline at the bottom boundary of the oceanic mixed layer. In formulating the theoretical model for this comparison, we expect to enable measurements of the sensitivity of surface wave propagation to its nonlinear coupling to the dynamics of the bathymetry, as well as the effect of thermal buoyancy, nonlinear dispersion and the complete Coriolis force. In addition, this work creates the theoretical methodology needed to include contemporary methodologies of stochastic uncertainty parameterisation which preserve the model's behaviour \cite{holm2015variational,ST2024,HH2021,ErwinThesis}.

\paragraph{Overview of the present work.} This paper extends the theoretical basis for simulations of mixed layer dynamics in $1.5$ layer models to determine the impacts of non-hydrostatic and thermal effects as well as non-traditional Coriolis forces. This is accomplished by deriving a version of the classic Green--Naghdi model which we extend to describe the mixed layer flow with long-wave dynamics of the upper free surface coupled to the dynamics of the thermocline layer at the bottom of the mixed layer, all riding on inert depths below.%
\footnote{This paper is based on unpublished notes from the mid 1990's during early work in the modelling phase of the COSIM (Computational Ocean and Sea Ice Model) project at Los Alamos National Laboratory. Some of the bibliography from that time is retained in this section as historical references whose presence might be helpful in tracing their influence in the development of modern methods of ocean mixed layer modelling during the past three decades.}

We expect that coupling the thermocline dynamics to horizontal gradients of the buoyancy in the mixed layer may have observable effects on the solution for the surface wave elevation. If so, these effects might be useful as diagnostics for the data produced by the NASA and ESA Surface Water and Ocean Topography (SWOT) satellite mission. 

\subsection{Some applications of the model and a description of the domain}

Historically, the Green-Naghdi (GN) model is a two dimensional model for a thin three dimensional fluid with a free surface. That is, through a vertical averaging procedure, it provides a description of two dimensional flow and wave dynamics within a `thin' fluid layer. The set up of the domain is therefore as follows. In the horizontal directions, we have coordinates $\bx =(x,y) \in \mathcal{D} \subseteq \mathbb{R}^2$ and in the vertical direction we have a coordinate $z \in [-H(x,y), \eta(x,y,t)]$, where $z=-H$ is the bottom boundary and $z=\eta$ is the free upper surface. Thus, the three dimensional domain is $\mathcal{D}\times[-H,\eta] \subset \mathbb{R}^3$. In this paper, we modify this setup such that the bottom boundary is permitted to have some exogenous time dependence. This setup is interesting for several physical reasons, including tsunami/wave generation and the ocean mixed layer. The derivation of the Green-Naghdi model for the mixed layer assumes a dynamically passive, hydrostatic lower layer. This assumption leads to the following relations  for partitioning the thickness of the mix-layer, denoted as $\eta(\bx,t)$, \cite{Gill1982,NMS}
\begin{align}\begin{split} 
  H(\bx,t) &= \rho_r h(\bx,t) - (1+\rho_r)B\,, \\
  \eta(\bx,t) &= (1-\rho_r) h(\bx,t) + (1+\rho_r)B
  \,. \label{bh}
\end{split}\end{align}
The symbol $B$ denotes the (constant) mean depth of the upper layer, $H(\bx,t)$ is the depth of the thermocline, and $\eta(\bx,t)$ is the height of the free surface,
using the undisturbed free surface as the reference level, at $\eta=0$.
The model parameter $\rho_r$ is the
ratio of the density of the upper layer to that of the lower layer.
Here, we assume $\rho_r = 1 - \epsilon$ with $\epsilon\ll1$, representing the small density difference 
of say 0.5\% across the thermocline.
It is clear from the relation $h=\eta+H$ in \eqref{bh} that most of the barotropic variation
in the surface layer depth $h(\bx,t)$ is driven by the movement of the thermocline 
$H(\bx,t)$; since the thermocline thickness $H$ in \eqref{bh} is more heavily weighted 
than the wave elevation, $\eta(\bx,t)$. 

\subsection{The historical Green-Naghdi Equations}

The historical Green-Naghdi equations are given by \cite{GN}
\begin{align}\begin{split} 
  \frac{\p \bu}{\p t}
   &=  - \,(\bu \cdot \nabla) \bu - g \nabla (h - H)
    + \frac{1}{h} \nabla \Big(h^2 \frac{dA}{dt}\Big) - \frac{dB}{dt} \nabla H \,,
  \\
  &\qquad \frac{\p h}{\p t} = - \nabla\cdot\, (h \bu) \,.
  \label{gn}
\end{split}\end{align}
Here, $\bu$ is the \emph{two dimensional} velocity field, $h=\eta(\bx,t)+H(\bx)$ is the local depth of the water,
$z=\eta(\bx,t)$ is the free surface height and $z=-H(\bx)$ is the
\emph{fixed} bottom boundary for the standard GN equations. Moreover, $\nabla$ denotes the two dimensional gradient in the $(x,y)$ coordinates and $d/dt = \p_t + \bu\cdot\nabla$ is the material derivative. The quantities $A$ and $B$ in \eqref{gn} are given by
\begin{align}\begin{split} 
  A & =   {\frac13} h \nabla\cdot \bu
                                   + {\frac12} \bu\cdot\nabla H  \,,
\\
  B & =   {\frac12} h \nabla\cdot \bu
                                   + \bu\cdot\nabla H  \,.
  \label{ab}
\end{split}\end{align}
The historical GN equations in \eqref{gn} provide a
depth-averaged (barotropic) description of shallow water motion with a
free surface under gravity, $g$, in {\it nondominant} asymptotics. That
is, the GN equations are derived using shallow water scaling
asymptotics for a domain of small aspect ratio, but no restriction is
placed on the Froude number, or nonlinearity parameter of the flow,
cf. reference \cite{CH1}.

The historical GN equations in \eqref{gn} have been derived by making a
solution Ansatz of columnar motion in which the horizontal fluid
velocity is independent of the vertical coordinate. One then imposes
three-dimensional incompressibility and certain symmetry requirements.
The GN equations are derived in \cite{GN} by requiring the incompressible
columnar motion to satisfy conservation of energy and invariance under
rigid body translations.  Equations were rediscovered in \cite{BMP} by inserting
the columnar motion solution Ansatz, $\p\bu/\p z=0$, and incompressibility
into the Euler equations, then averaging over depth. These equations have also 
been derived from a variational principle in \cite{MS} by inserting the
columnar motion Ansatz into the variational principle for the
Euler equations for an inviscid incompressible fluid, and explicitly
performing the vertical integrations before varying. The Hamiltonian
structure of these equations has been discussed in \cite{DH88}.

  The approach of Miles and Salmon \cite{MS} in obtaining the historical GN
equations by restricting to columnar motion in Hamilton's principle
for an incompressible fluid with a free surface affords a convenient
starting point for extending the GN equations to include effects of a
rotating frame and weak buoyancy stratification. Our derivation by
Hamilton's principle, following the derivation in \cite{HolmLuesink2021,ErwinThesis} inspired by Miles and Salmon, also provides a systematic derivation of the
Kelvin circulation theorem and potential vorticity conservation laws
for the extended GN equations. 

\section{An Euler-Poincar\'e derivation of the equations of motion}\label{GN}

In this section, we will augment the historical Green-Naghdi equations with the complete Coriolis force, thermal effects, and a time-varying lower boundary. We will derive the model by performing approximations in Hamilton's Principle. For further details about these approximations, we note that the nondimensional form of a stochastic perturbation of the historical Green-Naghdi equations with thermal effects was given in \cite{HolmLuesink2021}. Moreover, in \cite{SD2010}, the nontraditional Coriolis terms are discussed with reference to the Rossby number and aspect ratio.

We begin with the Lagrangian for the rotating Euler-Boussinesq equations in our three dimensional domain
\begin{equation}
    \ell_{EB} = \int_{\mathcal{D}\times[-H,\eta]} Dd^3x \bigg[ \underbrace{\frac12\left( |\bu|^2 + w^2 \right)}_{\textrm{Kinetic energy}} + \underbrace{\bR(\bx)\cdot\bu + S(\bx)w}_{\textrm{Coriolis terms}} - \underbrace{\rho gz}_{\substack{\text{Potential}\\\text{energy}}} - \underbrace{p(D^{-1}-1)}_{\substack{\text{Incompressibility}\\\text{in 3D}}}  \bigg] \,,
\end{equation}
where $Dd^3x \in {\rm Vol}(M)$ is the advected mass density and $\rho$ is an advected scalar thermal buoyancy. Denoting the three dimensional gradient by $\nabla_3$, the vector quantity
\begin{align}
2\bs{\Omega} = \nabla_3\times\bR(\bx) + \nabla_3 S(\bx)\times\zh
  \label{Coriolis}
\end{align}
is twice the rotation vector, $\bs{\Omega}$, which we note has a horizontal component when $S$ has a nonvanishing gradient. The domain is bounded below by bathymethy $z=-H(\bx,t)$ and the free surface $z=\eta(\bx,t)$, and we introduce the quantity $h = \eta+H$, which is the layer thickness. In a previous work \cite{HolmLuesink2021}, a similar variational method was used to derive thermal ocean models where, instead of the full thermal density $\rho$, the authors considered a stratified fluid by using a buoyancy variable $b$ related to $\rho$ through $\rho = \rho_0(1+b)$, where $\rho_0$ is some reference density. In this paper, we will instead consider the variable $\rho$ to be our advected variable. The effect of this change on the calculations to follow is minimal, as can be verified by comparing this paper with the Appendix of \cite{HolmLuesink2021}.

Since the domain has moving boundaries, care must be taken over the boundary conditions. On a moving surface defined by $F(x,y,z,t) = 0$, a particle remains on this surface if its Lagrangian differential is zero. That is, if $\p_t F + \bu\cdot\nabla F + w\p_z F = 0$. The standard kinematic boundary for free surfaces is recovered by setting $F = z-\eta(x,y,t)$, that is
\begin{equation}\label{eqn:kinematic_upper_boundary_condition}
    \p_t \eta + \bu|_{z=\eta}\cdot\nabla \eta = w|_{z=\eta} \,.
\end{equation}
Similarly, in this case the bottom boundary, $z=-H$, is permitted to move according to some exogenous dynamical processes. In this case, setting $F=z+H$ gives the condition that a particle on the bottom boundary remains on the bottom boundary, i.e.
\begin{equation}\label{eqn:kinematic_lower_boundary_condition}
    \p_t H + \bu|_{z=-H}\cdot\nabla H = -w|_{z=-H} \,.
\end{equation}

Since the three dimensional flow governed by the Euler-Boussinesq model is incompressible, we can express the vertical derivative of the vertical velocity in terms of the two dimensional gradient, $\nabla$, as
\begin{equation}
    \p_zw = - \nabla\cdot \bu \,.
\end{equation}
To determine the vertical velocity at some point $(x,y,z)$, we integrate this constraint from $z=-H$ to $z$ to give
\begin{equation}
\begin{aligned}
    w &= - \int_{-H}^z \nabla\cdot\bu dz' + w|_{z=-H} 
    \\
    &= - \nabla\cdot\int_{-H}^z \bu \,dz' + \bu|_{z=-H}\cdot\nabla H + w|_{z=-H} =  - \nabla\cdot\int_{-H}^z \bu \,dz' - \p_t H \,,
\end{aligned}
\end{equation}
where we have used a multidimensional version of Leibniz' integration rule and the boundary condition \eqref{eqn:kinematic_lower_boundary_condition}. To arrive at the Green-Naghdi model, we must introduce a columnar motion approximation. Put simply, this can be achieved by replacing the horizontal velocity field $\bu$ by its vertical average\footnote{We will, in general, denote vertically averaged variables using the 'overbar' notation.} $\ob\bu$. Since, the vertically averaged horizontal velocity depends only on $x,y,t$, we have that
\begin{equation}
\begin{aligned}
    w &= -\nabla\cdot\left( \ob\bu (z+H) \right) - \p_tH 
    \\
    &= -(z+H)\nabla\cdot\ob\bu - \ob\bu\cdot\nabla H - \p_t H = -(z+H)\nabla\cdot\ob\bu - \frac{dH}{dt} \,,
\end{aligned}
\end{equation}
where we have introduced the notation
\begin{equation}
    \frac{d}{dt} := \p_t + \ob\bu\cdot\nabla
\end{equation}
for the material derivative of a scalar quantity along a Lagrangian trajectory of the vertically averaged velocity field $\ob\bu$. The Lagrangian can be therefore approximated as
\begin{equation}
\begin{aligned}
    \ell_{GN} &= \int_{\mathcal{D}} dxdy \int_{-H(\bx,t)}^{\eta(\bx,t)} dz \bigg[ \frac{|\ob\bu|^2}{2} + \frac12 \left( (z+H)(\nabla\cdot\ob\bu) + \frac{dH}{dt} \right)^2 
    \\
    &\qquad\qquad\qquad\qquad\qquad + \bR(\bx)\cdot\ob\bu - S\left((z+H)(\nabla\cdot\ob\bu) + \frac{dH}{dt} \right) - \ob\rho g z \bigg]
    \\
    &= \int_{\mathcal{D}} dxdy \bigg[ \frac{h}{2}|\ob\bu|^2 + h \bR(\bx)\cdot\ob\bu + \frac16 h^3(\nabla\cdot\ob\bu)^2 + \frac12h^2 (\nabla\cdot\ob\bu)\frac{dH}{dt} + \frac{h}{2}\left(\frac{dH}{dt}\right)^2
    \\
    &\qquad\qquad\qquad - \frac12h^2 S(\bx)(\nabla\cdot\ob\bu) - h S(\bx)\frac{dH}{dt} - \frac12 \ob\rho (\eta^2-H^2) \bigg]
    \\
    &= \int_{\mathcal{D}} dxdy \,h\bigg[ \frac{|\ob\bu|^2}{2} + \bR(\bx)\cdot\ob\bu + \frac16 h^2(\nabla\cdot\ob\bu)^2 + \frac12h (\nabla\cdot\ob\bu)\left(\frac{dH}{dt} - S(\bx)\right) 
    \\
    &\qquad\qquad\qquad + \frac12\frac{dH}{dt}\left( \frac{dH}{dt} - 2S(\bx) \right) - \frac12\ob{\rho}g (h-2H) \bigg] \,,
\end{aligned}
\end{equation}
and this is the Lagrangian for our thermal rotating Green-Naghdi model with variable bathymetry. The equation of motion then follows from the Euler-Poincar\'e theorem, and takes the form
\begin{equation}
  \frac{d}{dt}\left( \frac{1}{h}\frac{\delta\ell_{GN}}{\delta\ob\bu} \right)
    + \frac{1}{h}\frac{\delta\ell_{GN}}{\delta \ob u^j}\nabla \ob u^j
    = \nabla \frac{\delta\ell_{GN}}{\delta h} - \frac{1}{h}\frac{\delta\ell_{GN}}{\delta\ob\rho}\nabla\ob\rho \,,
\label{eqn:EP_GN}
\end{equation}
where the right hand side of this equation results from the so-called `diamond terms' inherant to the semidirect product theory of continuum dynamics with advected quantities \cite{HMR1998}. Here, this is the symmetry broken Euler-Poincar\'e equation where the variables naturally live within a semidirect product between the diffeomorphisms on ${\rm Diff}(\mathcal{D})$ (or it's associated algebra of vector fields) and the space of advected quantities $\Omega^0(\mathcal{D})\oplus \Omega^2(\mathcal{D})$. The left hand side of the above equation is a coordinate expression of
\begin{equation}
    (\p_t + \mcal{L}_{\ob u})\left( \frac{1}{h}\frac{\delta\ell_{GN}}{\delta \ob u} \right) \,,
\end{equation}
where $\mathcal{L}_{\ob u}$ denotes the Lie derivative by the vector field $\ob{u}$.

The variational derivatives of the Lagrangian are computed as
\begin{equation}
\begin{aligned}
    \frac{\delta\ell_{GN}}{\delta h} &= \frac12|\ob\bu|^2 + \bR(\bx)\cdot\ob\bu + \frac12h^2(\nabla\cdot\ob\bu)^2 
    \\
    &\qquad + h(\nabla\cdot\ob\bu)\left( \frac{dH}{dt} - S(\bx) \right) + \frac12\frac{dH}{dt}\left( \frac{dH}{dt} - 2S(\bx) \right) - \ob\rho g(h-H)
    \,,\\
    \frac{\delta\ell_{GN}}{\delta\ob\rho} &= -\frac12 gh(h-2H)
    \,,\\
    \frac{1}{h}\frac{\delta\ell_{GN}}{\delta\ob\bu} &= \ob\bu + \bR(\bx)-\frac{1}{3h}\nabla\left( h^3(\nabla\cdot\ob\bu) \right) + \frac12 h(\nabla\cdot\ob\bu)\nabla H - \frac{1}{2h}\nabla\left( h^2\frac{dH}{dt} \right) + \frac{1}{2h}\nabla\left( h^2S(\bx) \right)
    \\
    &\qquad\qquad +\frac{1}{2}\frac{dH}{dt} \nabla H - S(\bx)\nabla H \\
    &= \ob\bu + \bR(\bx) - \frac{1}{h}\nabla(h^2F) + G\nabla H \,,
\end{aligned}
\end{equation}
where $F$ and $G$ are given by (cf. equation \eqref{ab})
\begin{equation}
\begin{aligned}
    F &= \frac13h\nabla\cdot\ob\bu + \frac12\left(\frac{dH}{ dt}-S(\bx)\right)
    \,,\\
    G &= \frac12h\nabla\cdot\ob\bu + \left( \frac{dH}{dt} - S(\bx) \right)
    \,.
\end{aligned}
\end{equation}

To make sense of the left hand side of the Euler-Poincar\'e equation \eqref{eqn:EP_GN}, we will make use of two equivalent forms of the Lie derivative of a vector field acting on a one-form, in two dimensional domains, which are
\begin{equation*}
    \mathcal{L}_{\ob \bu}(\bs{A}\cdot d\bx) = \left( \ob\bu\cdot\nabla\bs{A} + A_j\nabla\ob u^j \right)\cdot d\bx = \left( (\nabla^\perp\cdot \bs{A})\ob\bu^\perp + \nabla(\ob\bu\cdot\bs{A})\right)\cdot d\bx \,,
\end{equation*}
where we have introduced the perpendicular gradient, using the notation $(x,y)^\perp = (-y,x)$. Thus, after substituting in the variational derivatives, we have that the left hand side of the Euler-Poincar\'e momentum equation is given by
\begin{equation}\label{eqn:EP_LHS_1}
\begin{aligned}
    (\p_t + \mathcal{L}_{\ob u})\left( \frac{1}{h}\frac{\delta\ell_{GN}}{\delta \ob u} \right) &= \left( \p_t\ob\bu + \ob\bu\cdot\nabla\ob\bu + \ob u_j\nabla\ob u^j + (\nabla^\perp\cdot\bR)\ob\bu^\perp + \nabla(\ob\bu\cdot\bR)\right)\cdot d\bx
    \\
    &\qquad +(\p_t+\mathcal{L}_{\ob u})\left( \bigg(- \frac{1}{h}\nabla\left( h^2F \right) + G\nabla H \bigg)\cdot d\bx \right)\,,
\end{aligned}
\end{equation}
where $d\bx$ is the coordinate basis for one-forms. In the following calculations, we will make use of the fact that $h\,d^2x$ is advected, i.e.
\begin{equation}
    0 = (\p_t + \mathcal{L}_{\ob u})(h\,d^2x) = (\p_t h + \nabla\cdot(h\ob\bu) )d^2x = \left( \frac{dh}{dt} + h(\nabla\cdot\ob\bu) \right)d^2x \,.
\end{equation}
In the expression for $\frac{1}{h}\frac{\delta\ell_{GN}}{\delta \ob u}$, since we have divided through by the density, when $h$ appears it does so as a \emph{function} without its coordinate basis $d^2x$. Thus, we have that expanding out equation \eqref{eqn:EP_LHS_1} produces
\begin{equation}
\begin{aligned}
    (\p_t + \mathcal{L}_{\ob u})\left( \frac{1}{h}\frac{\delta\ell_{GN}}{\delta \ob u} \right) &= \bigg( \p_t\ob\bu + \ob\bu\cdot\nabla\ob\bu + \ob u_j\nabla\ob u^j + (\nabla^\perp\cdot\bR)\ob\bu^\perp + \nabla(\ob\bu\cdot\bR) - \frac{1}{h}\nabla\left( h^2\frac{dF}{dt} \right) 
    \\
    &\qquad + \frac{dG}{dt}\nabla H + \frac{1}{h^2}\frac{dh}{dt}\nabla(h^2F) - \frac{1}{h}\nabla\left( 2h \frac{dh}{dt} F \right) + G\nabla\frac{dH}{dt} \bigg)\cdot d\bx\,,
\end{aligned}
\end{equation}
where we have used that $\p_t + \mathcal{L}_{\ob u}$ obeys the standard rules of calculus and is $d/dt$ when it acts on the functions $F,G,h,H$. Making use of the explicit formulae for $F$ and $G$, as well as the fact that $h\,d^2x$ is advected as a volume form, equation \eqref{eqn:EP_LHS_1} can now be expressed as
\begin{equation*}
\begin{aligned}
    (\p_t + \mathcal{L}_{\ob u})\left( \frac{1}{h}\frac{\delta\ell_{GN}}{\delta \ob u} \right) &= \bigg( \p_t\ob\bu + \ob\bu\cdot\nabla\ob\bu + \ob u_j\nabla\ob u^j + (\nabla^\perp\cdot\bR)\ob\bu^\perp + \nabla(\ob\bu\cdot\bR) - \frac{1}{h}\nabla\left( h^2\frac{dF}{dt} \right) + \frac{dG}{dt}\nabla H
    \\
    &\qquad  -\frac{\nabla\cdot\ob\bu}{h}\nabla\left(h^2\left(\frac13h\nabla\cdot\ob\bu + \frac12\left(\frac{dH}{ dt}-S\right)\right)\right)
    \\
    &\qquad + \frac{1}{h}\nabla\left( 2h (h\nabla\cdot\ob\bu)\left( \frac13h\nabla\cdot\ob\bu + \frac12\left(\frac{dH}{ dt}-S\right)\right) \right) 
    \\
    &\qquad + \left(\frac12h\nabla\cdot\ob\bu + \left( \frac{dH}{dt} - S \right)\right)\nabla\frac{dH}{dt} \bigg)\cdot d\bx 
    \\
    &= \bigg( \p_t\ob\bu + \ob\bu\cdot\nabla\ob\bu + \ob u_j\nabla\ob u^j + (\nabla^\perp\cdot\bR)\ob\bu^\perp + \nabla(\ob\bu\cdot\bR) - \frac{1}{h}\nabla\left( h^2\frac{dF}{dt} \right) + \frac{dG}{dt}\nabla H
    \\
    &\qquad -\frac13(h\nabla\cdot\ob\bu)\nabla\left( h\nabla\cdot\ob\bu \right) - \frac23 (h\nabla\cdot\ob\bu)^2\frac{\nabla h}{h} - \frac{h\nabla\cdot\ob\bu}{2}\nabla\left( \frac{dH}{dt} - S \right)
    \\
    &\qquad - (\nabla\cdot\ob\bu)\left(\frac{dH}{dt} - S\right) \nabla h + \frac23 \nabla\left( (h\nabla\cdot\ob\bu)^2\right) + \frac23 (h\nabla\cdot\ob\bu)^2\frac{\nabla h}{h}
    \\
    &\qquad (h\nabla\cdot\ob\bu)\left(\frac{dH}{dt}-S\right)\frac{\nabla h}{h} + \nabla\left( (h\nabla\cdot\ob\bu)\left(\frac{dH}{dt}-S\right) \right) + \frac12(h\nabla\cdot\ob\bu)\nabla\frac{dH}{dt} 
    \\
    &\qquad + \left(\frac{dH}{dt} - S\right)\nabla\frac{dH}{dt} \bigg)\cdot d\bx \,,
\end{aligned}
\end{equation*}
where we have applied the product and chain rules to go from the first line to the second. Canceling terms and grouping terms which can be written as a gradient, we have that
\begin{equation*}
\begin{aligned}
    (\p_t + \mathcal{L}_{\ob u})\left( \frac{1}{h}\frac{\delta\ell_{GN}}{\delta \ob u} \right) &= \bigg( \p_t\ob\bu + \ob\bu\cdot\nabla\ob\bu + \ob u_j\nabla\ob u^j + (\nabla^\perp\cdot\bR)\ob\bu^\perp + \nabla(\ob\bu\cdot\bR) - \frac{1}{h}\nabla\left( h^2\frac{dF}{dt} \right) + \frac{dG}{dt}\nabla H
    \\
    &\qquad + \frac12\nabla\left( (h\nabla\cdot\ob\bu)^2 \right) + \nabla\left( (h\nabla\cdot\ob\bu)\left(\frac{dH}{dt}-S\right) \right) + \frac12\nabla\left(\frac{dH}{dt}\right)^2 
    \\
    &\qquad - S\nabla\frac{dH}{dt} + \frac{h\nabla\cdot\ob\bu}{2}\nabla S \bigg)\cdot d\bx
    \\
    &= \bigg( \p_t\ob\bu + \ob\bu\cdot\nabla\ob\bu + \ob u_j\nabla\ob u^j + (\nabla^\perp\cdot\bR)\ob\bu^\perp + \nabla(\ob\bu\cdot\bR) - \frac{1}{h}\nabla\left( h^2\frac{dF}{dt} \right) + \frac{dG}{dt}\nabla H
    \\
    &\qquad + \nabla\left( \frac12(h\nabla\cdot\ob\bu)^2 + (h\nabla\cdot\ob\bu)\left( \frac{dH}{dt} - S \right) + \frac12 \frac{dH}{dt}\left( \frac{dH}{dt} - 2S\right)  \right) 
    \\
    &\qquad + \left( \frac{dH}{dt} + \frac{h\nabla\cdot\ob\bu}{2} \right)\nabla S \bigg)\cdot d\bx \,.
\end{aligned}
\end{equation*}
The right hand side of the Euler-Poincar\'e equation \eqref{eqn:EP_GN} is
\begin{equation}
\begin{aligned}
    \nabla \frac{\delta\ell_{GN}}{\delta h} - \frac{1}{h}\frac{\delta\ell_{GN}}{\delta\ob\rho}\nabla\ob\rho &= \nabla\bigg( \frac12|\ob\bu|^2 + \bR(\bx)\cdot\ob\bu + \frac12h^2(\nabla\cdot\ob\bu)^2 + h(\nabla\cdot\ob\bu)\left( \frac{dH}{dt} - S \right) 
    \\
    &\qquad + \frac12\frac{dH}{dt}\left( \frac{dH}{dt} - 2S \right) - \ob\rho g(h-H) \bigg) + \frac12g(h-2H)\nabla\ob\rho
    \\
    &=\nabla\left( \frac{|\ob\bu|^2}{2} + \bR\cdot\ob\bu \right) - \frac12gh\nabla\ob\rho - g\ob\rho \nabla(h-H) 
    \\
    &\qquad + \nabla\left( \frac12(h\nabla\ob\bu)^2 + (h\nabla\cdot\ob\bu)\left(\frac{dH}{dt}-S\right) + \frac12\frac{dH}{dt}\left(\frac{dH}{dt}-2S\right) \right) \,.
\end{aligned}
\end{equation}
Combining these calculations finally gives the following extended Green-Naghdi equations
\begin{align}
    \frac{d\ob\bu}{dt} + (\nabla^\perp\cdot\bR)\ob\bu^\perp &= -g\ob\rho \nabla(h-H) - \frac12gh\nabla\ob\rho + \frac{1}{h}\nabla\left(h^2\frac{dF}{dt} \right) - \frac{dG}{dt}\nabla H - \left( \frac{dH}{dt} + \frac{h\nabla\cdot\ob\bu}{2} \right)\nabla S
    \label{eqn:extended_GN_1}
    \,,\\
    \frac{d\ob\rho}{dt} &= 0
    \label{eqn:extended_GN_2}
    \,,\\
    \frac{dh}{dt} &= - h\nabla\cdot\ob\bu
    \label{eqn:extended_GN_3}
    \,.
\end{align}

\section{Some features of the model}

The equations presented in the previous section have a potential vorticity formulation however, since the model includes thermal effects, this potential vorticity is not conserved. The most direct way to see this is through Kelvin's circulation theorem, which follows from the Euler-Poincar\'e formulation of the model as follows. To express Kelvin's circulation theorem, we first must introduce the notion of closed continuous loops within the fluid, $c:C^1\rightarrow \mathcal{D}$, the space of which is denoted by $\mathcal{C}(\mathcal{D})$. The group ${\rm Diff}(\mathcal{D})$ acts on $\mathcal{C}(\mathcal{D})$ from the left by composition, that is, $\phi c := \phi \circ c : C^1\rightarrow \mathcal{D}$ is a group action for $\phi \in {\rm Diff}(M)$ and $c \in \mathcal{C}(\mathcal{D})$. For an initial loop of fluid $c_0 \in \mathcal{C}(\mathcal{D})$, if the flow map of the fluid $\phi_t \in {\rm Diff}(\mathcal{D})$ acts on this initial loop then we obtain a closed loop moving with the flow $c_t = \phi_tc_0$.

\begin{proposition}
    For a closed loop, $c_t \in \mathcal{C}(\mathcal{D})$, enclosing a region of fluid, $S_t \subset \mathcal{D}$, and moving with the flow of the Green-Naghdi equations \eqref{eqn:extended_GN_1} - \eqref{eqn:extended_GN_3}, we have the following Kelvin circulation relation
    \begin{equation}\label{eqn:Kelvin_statement}
        \frac{d}{dt}\oint_{c_t} \left(\ob\bu + \bR(\bx) - \frac{1}{h}\nabla(h^2F) + G\nabla H \right)\cdot d\bx = \frac{g}{2}\int_{S_t}\nabla^\perp(h - 2H)\cdot\nabla\ob\rho \,dS \,.
    \end{equation}
\end{proposition}
\begin{proof}
    The proof of this follows from the transport theorem for the contour $c_t$ moving with the fluid velocity $\ob\bu$, and the form of the momentum equation presented in equation \eqref{eqn:EP_GN}. Indeed,
    \begin{align}\begin{split} 
        \frac{d}{dt}\oint_{c_t}
        \left( \frac{1}{h}\frac{\delta\ell_{GN}}{\delta\ob\bu} \right)\cdot d\bx
        &=  \oint_{c_t}
        \left[\frac{d}{dt}\left( \frac{1}{h}\frac{\delta\ell_{GN}}{\delta\ob\bu} \right) + \frac{1}{h}\frac{\delta\ell_{GN}}{\delta \ob u^j}\nabla \ob u^j\right]\cdot d\bx
        \\
        &
        = \oint_{c_t}\left[\nabla \frac{\delta\ell_{GN}}{\delta h} - \frac{1}{h}\frac{\delta\ell_{GN}}{\delta\ob\rho}\nabla\ob\rho \right]\cdot d\bx = \oint_{c_t} \frac12 g(h-2H)\nabla\ob\rho \cdot d\bx \,,
        \label{kel}
\end{split}\end{align}
    and Green's theorem then gives the required result.
\end{proof}
Notice that, for this model, we have that circulation is created whenever the gradient of the bathymetry or surface elevation does not align with the gradient of the thermal buoyancy. The Kelvin circulation result \eqref{eqn:Kelvin_statement} permits us to understand the potential vorticity dynamics. In particular, we define the potential vorticity by
\begin{equation}
    q = \frac{1}{h}\nabla^\perp\cdot\left(\ob\bu + \bR(\bx) - \frac{1}{h}\nabla(h^2F) + G\nabla H \right) \,,
\end{equation}
and the dynamics of $q$ are then governed by
\begin{equation}\label{eqn:PV_dynamics}
    \p_tq + \ob\bu\cdot\nabla q = \frac{g}{2}\frac{1}{h}\nabla^\perp(h - 2H)\cdot\nabla\ob\rho \,.
\end{equation}
In the definition of $q$ given above, it has been expressed in terms of the perpendicular two dimensional gradient to ensure that the equations are expressed in terms of a self consistent two dimensional vector calculus convention. Reintroducing the third coordinate in the vertical direction, we note that this is equivalent to $q = \frac{1}{h}\zh\cdot\nabla_3\times\frac{1}{h}\frac{\delta\ell_{GN}}{\delta\ob\bu}$.
\begin{remark}[The influence of thermal effects]
    The presence of thermal inhomogeneity in the model has resulted in forcing terms on the right hand side of the Kelvin circulation statement \eqref{eqn:Kelvin_statement} and on the potential vorticity equation \eqref{eqn:PV_dynamics}. When the advected quantitiy $\ob\rho$ is not present (i.e. when it is constant), these terms vanish and the potential vorticity is purely advected. Thus, in the case without thermal effects, the equations have infinitely many conserved quantities (Casimirs) which are the arbitrary functions of $q$ integrated against the measure $h dxdy$. 
\end{remark}
The equations presented in this paper, complete with thermal effects, also possess infinitely many conserved integral quantities. Corresponding to suitably well-behaved functions $\Phi,\Psi$ of the thermal buoyancy $\ob\rho$, we have the following Casimir conserved quantity
\begin{equation}\label{eqn:Casimirs}
    C_{\Phi,\Psi} = \int_{\mathcal{D}} \left(\Phi(\ob\rho) + q\Psi(\ob\rho)\right) h\,dxdy \,.
\end{equation}
There are infinitely many such quantities, inherited from the dimension of the space of suitable functions of $\ob\rho$. We here note that, in the absense of thermal effects, the Casimirs for the standard Green-Naghdi (or shallow water) equations are dramatically different. That is, potential vorticity is an advected variable when $\ob\rho$ is a constant and, as is known for standard rotating shallow water equations \cite{DS05}, the Casimir invariants would be simply arbitrary functions of potential vorticity integrated against the depth weighted measure. In the symmetry broken thermal fluid model considered here, the form of the Casimirs is given in \eqref{eqn:Casimirs} and is standard for such semidirect product systems (see e.g. \cite{Dellar2003, HHS2025}) and was also presented for the Green-Naghdi equations in \cite{ErwinThesis}. These Casimirs have the same structure as those for the thermal shallow water model (also known as the ${\rm IL}^0$ equations) \cite{BV21a}.

\section{Comments and outlook}

In this paper, we introduce an extension to the standard Green-Naghdi model for vertically averaged geophysical flows. In particular, we introduce a complete Coriolis force \cite{DS05}, thermal effects, and a time-dependent bathymetry. This time dependent bathymetry makes the model suitable to be used as the topmost layer of a multilayer ocean model, in which the bottom boundary is taken to be the thermocline or a free boundary between layers with different modelling assumptions. This model has been derived using the Euler-Poincar\'e approach, and is therefore intrinsically geometric in nature. This has permitted the formualation of Kelvin's circulation theorem, a potential vorticity formulation, and infinitely many `Casimir' conserved quantities. In a future work, it may be of interest to introduce the nondimensional form of the equations and illuminate the relationship between the model derived here and a family of related models. Such a family is well known for the standard case for which the bottom boundary is a fixed function of space.

This work has application towards other geophysical problems for which a changing bottom boundary is relevant. In particular, this work is of relevance to wave generation by landslides or earthquakes \cite{FB2024}, a problem for which accurately accounting for the time dependence in the bathymetry is crucial for accurate forecasting. The equations presented here are deterministic in nature, and follow from classical geometric approaches to model derivation. It is possible, following e.g. \cite{holm2015variational,ST2024,HH2021,ErwinThesis}, to include stochastic terms through the geometric structure of the model for the purposes of data assimilation and uncertainty quantification. There is particular interest in interpreting data in this field due to its application to tsunami forecasting and, in particular, the assimilation of experimental wave data (see  e.g. \cite{WND2015}) into a stochastic Green-Naghdi model is an interesting proposition. Testing this proposal is beyond the scope of the present paper.

\paragraph{Acknowledgments.} We would like to thank our friends R. Hu, C. Cotter, J. Woodfield, and E. Luesink for their helpful discussions on geophysical fluids during the course of this work and beyond. Oliver D. Street acknowledges funding for a research fellowship from Quadrature Climate Foundation. Darryl D. Holm was partially supported during the present work by European Research Council (ERC) Synergy grant Stochastic Transport in Upper Ocean Dynamics (STUOD) – DLV-856408, as well as partially supported by Office of Naval Research (ONR) grant award N00014-22-1-2082, Stochastic Parameterization of Ocean Turbulence for Observational Networks.

\bibliographystyle{alpha}
\bibliography{main}

\end{document}